\documentclass[11pt]{article}

\usepackage[a4paper,margin=2cm]{geometry}

\usepackage{amsmath, amsfonts}

\usepackage{amsthm}
\newtheorem{theorem}{Theorem}
\newtheorem{lemma}{Lemma}

\newtheorem{proposition}{Proposition}

\newtheorem{corollary}{Corollary}
\theoremstyle{definition}

\theoremstyle{remark}
\newtheorem{remark}{Remark}

\usepackage{xcolor}

\usepackage{subcaption}
\usepackage{tikz}
\usetikzlibrary{decorations.pathreplacing}
\tikzset{
    vertex/.style={circle, draw, minimum size=7mm},
    offvertex/.style={rectangle, draw, minimum size=7mm},
    edge/.style={draw, thick},
}

\usepackage{enumerate}

\usepackage[ruled, noend, linesnumbered]{algorithm2e}
\DontPrintSemicolon

\usepackage[draft]{changes}
\definecolor{darkgreen}{rgb}{0,0.7,0}
\definechangesauthor[name={reviewer}, color=darkgreen]{rev}
\definecolor{darkred}{rgb}{0.8,0,0}
\definechangesauthor[name={yongho}, color=darkred]{yshin}
\setcommentmarkup{\textcolor{authorcolor}{[#1]}}

\usepackage{hyperref}

\DeclareMathOperator{\E}{\mathbb{E}}
\DeclareMathOperator{\Unif}{\mathrm{Unif}}

\newcommand{\dd}{\mathrm{d}}
\newcommand{\e}{\mathrm{e}}

\newcommand{\R}{\mathbb{R}}

\newcommand{\calC}{\mathcal{C}}
\newcommand{\calA}{\mathcal{A}}
\newcommand{\calB}{\mathcal{B}}
\newcommand{\calE}{\mathcal{E}}
\newcommand{\calD}{\mathcal{D}}

\newcommand{\deltared}{\delta^{(\mathrm{red})}}
\newcommand{\deltaeff}{\delta^{(\mathrm{eff})}}

\title{Online Rounding Schemes for Edge Cover}
\date{}

\author{Jarosław Byrka\footnote{Institute of Computer Science, University of Wrocław. Email: \texttt{jby@cs.uni.wroc.pl}.} \and Yongho Shin\footnote{Institute of Computer Science, University of Wrocław. Email: \texttt{yongho@cs.uni.wroc.pl}.}}

\begin{document}
\maketitle
\begin{abstract}
\emph{Linear programming rounding} is a fundamental technique in the design of approximation and online algorithms. 
A growing body of work has investigated online rounding schemes for various combinatorial optimization problems, with applications including multi-stage stochastic optimization. 
More recently, Naor, Srinivasan, and Wajc (SODA 2025) study online rounding for bipartite ($b$-)matching and achieve a rounding factor better than $1 - 1/\e$.

Motivated by this line of research, we study online rounding schemes for the \emph{edge cover} problem. 
Our main contribution is a $1.8$-competitive online rounding scheme for general (possibly non-bipartite) graphs. 
The key idea is to incorporate a controlled form of adaptivity based on the coverage status of endpoints, while guaranteeing both upper and lower bounds on the probability that each edge is selected. 
We further obtain an improved $(1 + 1/\e)$-competitive online rounding scheme for bipartite graphs under a model in which all edges incident to an ``online'' vertex are revealed simultaneously in each round. 
Finally, we establish lower bounds on the competitive ratio for bipartite graphs, demonstrating a strict separation between offline and online rounding schemes.
\end{abstract}


\section{Introduction}

\emph{Linear programming (LP) rounding} is a technique that takes as input a feasible fractional solution to an LP relaxation and rounds it into an integral solution while incurring only a small loss in objective value. 
This technique is a fundamental tool in the design of approximation algorithms and has been successfully applied to a wide range of NP-hard combinatorial optimization problems, often yielding strong approximation guarantees. 
For a comprehensive introduction to rounding-based approximation algorithms, we refer the reader to the textbook of Williamson and Shmoys~\cite{williamson2011design}.

To address uncertainty in future inputs, rounding schemes have also been studied in an \emph{online arrival} model, where the entries of the fractional solution are revealed sequentially in an adversarial order, and the scheme must make an irrevocable decision for each element upon its arrival.
A prominent application of such online rounding schemes arises in \emph{multi-stage stochastic covering}~\cite{shmoys2006approximation,swamy2012sampling}. 
In this setting, ground elements are revealed over multiple stages according to a known distribution, while the costs of subsets may increase over stages. 
The goal is to construct a solution that covers all realized elements while minimizing the expected total cost. 
A standard approach~\cite{shmoys2006approximation,swamy2012sampling,byrka2018approximation} first computes a near-optimal fractional solution to an LP relaxation based on the known distribution and then rounds it into an integral solution. 
Crucially, the rounding procedure herein must make irrevocable decisions at each stage, since delaying decisions may lead to higher costs. 
Online rounding schemes are well suited to capture this requirement.

Recent progress has also been made for online rounding in \emph{bipartite ($b$-)matching} by Naor, Srinivasan, and Wajc~\cite{naor2025online}, with applications such as online edge coloring and the learning-augmented setting. 
They consider a \emph{one-sided vertex arrival} model such that, upon arrival of an online vertex and the fractional solution values on its incident edges, the scheme immediately and irrevocably decides which edge to match, with a guarantee that the input fractional solution is a feasible matching in the hindsight graph.
They present schemes with rounding factor $0.652$ for online matching and $0.646$ for $b$-matching, which are strictly better than the classical barrier of $1-1/\e$ in online bipartite matching.
They also establish an upper bound of $0.828$, exhibiting a clear separation between offline and online rounding schemes for bipartite matching.

Motivated by this line of work, we study online rounding schemes for the \emph{edge cover} problem. 
Given a (possibly non-bipartite) graph with edge costs, the goal of the problem is to find a minimum-cost set of edges such that every vertex is incident to at least one selected edge. 
In an online rounding scheme for edge cover, the edges are sequentially revealed along with their corresponding fractional values, and the scheme must irrevocably decide whether to include each edge upon arrival, with the guarantee that the fractional values constitute a feasible fractional edge cover in hindsight.

A natural baseline rounding scheme for this problem is as follows. 
For each vertex, it independently samples a threshold uniformly at random from $[0,1)$. 
When an edge $e$ arrives, it selects $e$ if there exists an endpoint $u \in e$ such that the cumulative fractional coverage of $u$ exceeds its threshold for the first time. 
This  rounding scheme is $2$-competitive (i.e., it achieves a factor of $2$ in expected cost), and in fact, this analysis is tight due to a large complete graph with the evenly distributed fractional solution.
This naturally raises the question of whether one can design online rounding schemes that achieve a competitive ratio strictly better than 2.

In this paper, we answer this question in the affirmative by presenting improved online rounding schemes for the edge cover problem. 
For \emph{general graphs} under the \emph{edge arrival} model, we obtain the following result.

\begin{theorem} \label{thm:intro:edgecvr}
    There exists a $1.8$-competitive online rounding scheme for edge cover.
\end{theorem}

To provide some intuition on this rounding scheme, observe that the simple $2$-competitive scheme is \emph{non-adaptive}: it selects an edge whenever the threshold condition is satisfied by any endpoint, regardless of whether that endpoint has already been covered by a previously selected edge. 
This suggests that improved guarantees may be achievable by incorporating adaptivity of conditioning decisions on the current coverage status of the endpoints. 
However, we find that a fully adaptive scheme is difficult to analyze.
In particular, it becomes challenging to obtain lower bounds on the probability that a vertex is covered, since such bounds would require a lower bound on the probability that certain edges are selected, which is an objective that runs counter to the direction of the entire analysis. 
We overcome this difficulty by introducing a limited form of adaptivity, combined with a carefully designed threshold distribution for each vertex.

We further investigate \emph{bipartite graphs} under the \emph{one-sided vertex arrival} model studied by Naor et al.~\cite{naor2025online}.
In this model, all edges incident to an ``online'' vertex at each round are revealed simultaneously along with their fractional values.

\begin{theorem} \label{thm:intro:varr}
    There exists a $(1 + 1/\e)$-competitive online rounding scheme for edge cover in bipartite graphs under one-sided vertex arrivals.
\end{theorem}

Finally, we establish lower bounds on the competitive ratio of online rounding schemes. 
For general graphs, a standard instance given by a length-3 cycle immediately yields a lower bound of $4/3$. 
In contrast, although the standard LP relaxation is integral for bipartite graphs, we show that online rounding still incurs a nontrivial loss: any online rounding scheme must have competitive ratio at least $7/6$ under the edge arrival model and at least $9/8$ under the one-sided vertex arrival model. 
This demonstrates a strict separation between offline and online rounding schemes for bipartite graphs.

\paragraph{Related Work}
As a generalization of edge cover, the \emph{set cover} problem has also been extensively studied through the lens of rounding techniques in both offline~\cite{kolliopoulos2005approximation,srinivasan2006extension,buchbinder2018simplex,ma2023order} and online settings~\cite{alon2003online,buchbinder2009online,buchbinder2014competitive,byrka2018approximation}. 
More broadly, online rounding schemes have attracted increasing attention due to their success in a variety of online/stochastic optimization problems, including online/stochastic matching~\cite{gamlath2019online,fahrbach2022edge,blanc2022multiway,gao2022improved,shin2021making,chen2024stochastic,buchbinder2023lossless,naor2025online}, online edge coloring~\cite{cohen2019tight,saberi2021greedy,kulkarni2022online,blikstad2024online,blikstad2025deterministic}, and online scheduling~\cite{lattanzi2020online}. 
We also refer the reader to the literature on \emph{online contention resolution schemes}~\cite{feldman2021online,lee2018optimal,fu2022oblivious,zhao2025universal} and a wide range of their applications in Bayesian/stochastic optimization therein.

\section{Preliminaries}
In the \emph{edge cover} problem, given a graph $G = (V, E)$ with edge cost $c : E \to \R_{\geq 0}$, the goal is to find a subset of edges $C \subseteq E$ that covers every vertex (i.e., for every $v \in V$, there exists $e \in C$ such that $v$ is an endpoint of $e$) at minimum total cost $\sum_{e \in C} c(e)$.
Recall that the standard LP relaxation for edge cover is formulated as follows:
\begin{align*}
    \text{minimize } & \sum_{e \in E} c(e) \, x_e \\
    \text{subject to } & \sum_{e \in \delta(v)} x_e \geq 1, & \forall v \in V, \\
    & x_e \geq 0, & \forall e \in E,
\end{align*}
where $\delta(v)$ denotes the set of edges incident to $v$.
A \emph{rounding scheme} is defined to be a (possibly randomized) algorithm that takes as input a feasible fractional solution $x \in \R_{\geq 0}^E$ to the LP relaxation, together with $G = (V, E)$ and $c$, and outputs a feasible integral edge cover $C \subseteq E$.
We may further assume that $x \in [0, 1]^E$ since we can easily reduce every $x_e > 1$ to $1$ without violating the feasibility.

In this paper, however, we consider an \emph{online arrival} model for rounding schemes, defined as follows.
We may assume that the vertices $V$ are given upfront; however, the edges $E$ are revealed one at a time (or in batches) in an online manner.
When an edge $e \in E$ is revealed, the rounding scheme then learns its solution value $x_e$, as well as its cost $c_e$.
Moreover, the scheme must decide to select $e$ into its solution $C$ or not at this moment.
This decision is irrevocable: if the scheme decides to select $e$ into $C$ upon the arrival of $e$, it is not allowed to discard $e$ from $C$ in later time points; otherwise, the scheme cannot insert $e$ into $C$ in later time points, either.
We remark that, in this paper, we consider an oblivious adversary setting where $E$ is determined at the very beginning of execution.

Under this online arrival model, we say the rounding scheme is \emph{$\rho$-competitive} for some $\rho \geq 1$ when the expected total cost of the rounding scheme's final solution is within a factor $\rho$ of the cost of the fractional solution $x$, i.e., 
\[
    \E \Bigg[ \sum_{e \in C} c(e) \Bigg] \leq \rho \cdot \sum_{e \in E} c(e) \, x_e,
\]
where the expectation is over the randomness of the rounding scheme.
In fact, our rounding schemes achieve $\rho$-competitiveness by guaranteeing
that each edge $e$ is selected into its solution $C$ within a factor $\rho$ of its solution value, i.e.,
\begin{equation} \label{eq:pre:peredge}
    \Pr [e \in C] \leq \rho \cdot x_e \text{ for every $e \in E$.}
\end{equation}
Note that linearity of expectation then implies $\rho$-competitiveness.

As the edges are sequentially fed to the scheme, for any two edges $e$ and $e'$, we denote by $e \prec e'$ if $e$ precedes $e'$ in the arrival order.
Moreover, for each $v \in V$ and $e \in E$, let $\delta_e(v)$ denote the edges that are incident with $v$ and precede $e$, i.e., $\delta_e(v) = \delta(v) \cap \{e' \in E \mid e' \prec e\}$.

\section{1.8-Competitive Rounding Scheme} \label{sec:edge}
In this section, we present a $1.8$-competitive rounding scheme for general graphs.
In particular, we will show that the rounding scheme satisfies \eqref{eq:pre:peredge} with $\rho \leq 1.8$.

In this model, we can assume without loss of generality that exactly one edge arrives in each round.
Moreover, let $\tau > 0$ be a sufficiently small constant to be chosen later, and let us assume that $x_e \leq \tau$ for every $e \in E$.
To have this assumption without loss of generality, let us allow $G$ to be a multigraph.
We then split the edges into copies each of which has its fractional solution value no greater than $\tau$ and regard as these copies arrive one at a time in an arbitrary order.
Note that this process blows up the input size by a factor $O(\frac{1}{\tau})$.

We say an edge $e$ \emph{effectively covers} its endpoint $v$ if $v$ has not been fully covered when $e$ arrives, i.e., $\sum_{e' \in \delta_e(v)} x_{e'} < 1$; otherwise, we say $e$ \emph{redundantly covers} $v$.
Let $\deltaeff(v) \subseteq \delta(v)$ denote the set of edges effectively covering $v$, and let $\deltared(v) = \delta(v) \setminus \deltaeff(v)$ denote the set of edges redundantly covering $v$.
We may also write $\deltaeff_e(v) := \deltaeff(v) \cap \{e' \in E \mid e' \prec e\}$ and $\deltared_e(v) := \deltared(v) \cap \{e' \in E \mid e' \prec e\}$ to denote the edges that effectively and redundantly cover $v$, respectively, among those preceding $e$.

In the description and analysis of our rounding scheme, we further assume that every vertex $v$ is ``tightly'' covered by its effectively covering edges, i.e.,
\begin{equation} \label{eq:ec:assume:effcovedge}
    \sum_{e \in \deltaeff(v)} x_e = 1 \text{ for every $v \in V$}.
\end{equation}
Note that we can easily make this assumption without loss of generality by splitting the last edge $e$ in $\deltaeff(v)$ into two edges $e_1$ and $e_2$, arriving in this order with $x_{e_1} := 1 - \sum_{e' \in \deltaeff_e(v)} x_{e'}$ and $x_{e_2} := x_e - x_{e_1}$, for every vertex $v$ violating the assumption.

\paragraph{Scheme Description}
Let $p \in (0, 1]$ be a parameter; specifically, we choose $p := 0.46$.
Let $c := \frac{\e^{p^2} - 1}{p} + 1 > 1$, and let $g : [0, 1] \to \R_{\geq 0}$ be a function defined as follows:
\begin{equation} \label{eq:edcv:dist}
    g(z) := \frac{\ln c}{p} \cdot \frac{1}{p + (1-p) \cdot c^{-z}}.
\end{equation}
We highlight that $g(z) \geq 0$ for all $z \in [0, 1]$ and $\int_0^1 g(z) \dd z = 1$, implying that $g$ is a feasible probability density function.
Note also that $g$ is increasing.

For each vertex $v \in V$, the scheme samples $Z_v \in [0, 1]$ independently and identically from the distribution whose probability density function is $g$.
Let $C \subseteq E$ be the solution maintained by the scheme, initially $C := \emptyset$.
Upon arrival of an edge $e \in E$ with $x_e \leq \tau$, let $\ell_v$ be the fraction that $v \in e$ is covered until now, i.e., $\ell_v = \sum_{e' \in \delta_e (v)} x_{e'}$. 
For each endpoint $v$ effectively covered by $e$ (where $\ell_v = \sum_{e' \in \deltaeff_e(v)} x_{e'}$), if $\ell_v < Z_v \leq \ell_v + x_e$, $v$ marks $e$ unconditionally with probability $p$, and only when $v$ is not covered by $C$ with probability $1-p$.
If $e$ effectively covers exactly one endpoint (and hence redundantly covers the other endpoint), $e$ is marked by itself independently with probability $p \, g(0) \, x_e = x_e \, \ln c$.
The scheme then inserts $e$ into $C$ if and only if $e$ is marked.
We provide in Algorithm~\ref{alg:edcv} a pseudocode of this rounding scheme.

\begin{algorithm}[t]
    \caption{1.8-competitive rounding scheme} \label{alg:edcv}
    \KwIn{$G = (V, E)$ and $x \in [0, \tau]^E$}
    \KwOut{An edge cover $C \subseteq E$}
    $C \gets \emptyset$, $p \gets 0.46$ \;
    Let $\calD$ be the distribution on $[0, 1]$ whose p.d.f.\ is $g$ defined in \eqref{eq:edcv:dist}\;
    
    \For{each vertex $v \in V$}{
        Sample $Z_v \in [0, 1]$ independently from $\calD$ \;
    }
    \For{each edge $e \in E$, along with $x_e \leq \tau$}{
        \For{each endpoint $v$ effectively covered by $e$}{
            $\ell_v \gets \sum_{e' \in \deltaeff_e(v)} x_{e'}$\;
            \uIf{$\ell_v < Z_v \leq \ell_v + x_e$}{
                \textbf{with} probability $p$\;
                \Indp Mark $e$ by $v$ \;
                \Indm\textbf{otherwise}\;
                \Indp Mark $e$ by $v$ only if $v$ is uncovered \;
            }
        }
        \uIf{$e$ effectively covers exactly one endpoint} {
            Mark $e$ by itself with probability $p \, g(0) \, x_e$ independently \;
        }
        \uIf{$e$ is marked} {
            $C \gets C \cup \{e\}$ \;
        }
    }
    \Return{$C$}\;
\end{algorithm}

\paragraph{Analysis}
Observe that the scheme returns a feasible edge cover.
Indeed, for every $v \in V$, at the moment when $\ell_v < Z_v \leq \ell_v + x_e$ is satisfied for some $e \in \deltaeff(v)$, $v$ marks $e$ if $v$ is yet uncovered.
For the competitive ratio of this rounding scheme, we show the following key technical lemma.

\begin{lemma} \label{lem:ec:main}
    For any constant $\varepsilon > 0$, there exists a constant $\tau > 0$ such that the rounding scheme satisfies that, for every $e \in E$ and $v \in e$ effectively covered by $e$,
    \[
    \Pr[v \text{ marks } e] \leq \left(\frac{\ln c}{p} + \frac{\varepsilon}{2} \right) \, x_e.
    \]
\end{lemma}

\begin{proof}
Let $\ell_v$ be the value of $\ell_v$ at the moment of the arrival of $e$.
Let $N_e(v)$ be the neighbors of $v$ before $e$ arrives, i.e., $N_e(v) = \{ u \in V \mid \exists (u, v) \in \delta_e(v) \}$.
We remark that, for every $u \in N_e(v)$, $\delta_e(u) \cap \delta_e(v)$ contains all the parallel copies of edge $(u, v)$ that have arrived so far.
Moreover, as $v$ is effectively covered by $e$, we have $\delta_e(v) = \deltaeff_e(v)$, $\ell_v = \sum_{u \in N_e(v)} \sum_{e' \in \delta_e(u) \cap \deltaeff_e(v)} x_{e'}$, and $\ell_v + x_e \leq 1$.

Observe that $v$ marks $e$ if and only if
\begin{enumerate}
    \item \label{ev:ec:m1} $\ell_v < Z_v \leq \ell_v + x_e$, and
    \item \label{ev:ec:m2} either one of the following events is true:
    \begin{enumerate}
        \item \label{ev:ec:m2a} $v$ unconditionally marks $e$, or
        \item \label{ev:ec:m2b} for all $u \in N_e(v)$, any of $\deltaeff_e(u) \cap \deltaeff_e(v)$ is not marked by $u$, and any of $\deltared_e(u) \cap \deltaeff_e(v)$ is not marked by itself.
    \end{enumerate}
\end{enumerate}
Note that, when we define $\calE$ as the event such that
\begin{enumerate}
    \item \label{ev:ec:n1} $\ell_v < Z_v \leq \ell_v + x_e$, and
    \item \label{ev:ec:n2} either one of the following events is true:
    \begin{enumerate} [(i)]
        \item \label{ev:ec:n2a} $v$ unconditionally marks $e$, or
        \item \label{ev:ec:n2b} for all $u \in N_e(v)$, any of $\deltaeff_e(u) \cap \deltaeff_e(v)$ is not unconditionally marked by $u$, and any of $\deltared_e(u) \cap \deltaeff_e(v)$ is not marked by itself,
    \end{enumerate}
\end{enumerate}
we have $\Pr[v \text{ marks } e] \leq \Pr[\calE]$ since Event~\eqref{ev:ec:m2b} implies Event~\eqref{ev:ec:n2b}.
For every $u \in N_e(v)$ let $\calB_u$ denote the event that any of $\deltaeff_e(u) \cap \deltaeff_e(v)$ is not unconditionally marked by $u$, and any of $\deltared_e(u) \cap \deltaeff_e(v)$ is not marked by itself.
That is, Event~\eqref{ev:ec:n2b} is equivalent with $\bigwedge_{u \in N_e(v)} \calB_u$.

Let us focus on bounding $\Pr[\calB_u]$ for each $u \in N_e(v)$.
Observe that marking an edge $e' \in \deltaeff_e(u) \cap \deltaeff_e(v)$ is a disjoint event from marking another $e'' \in \deltaeff_e(u) \cap \deltaeff_e(v) \setminus \{e'\}$.
Moreover, for any $e' \in \deltaeff_e(u) \cap \deltaeff_e(v)$, when denoting by $\ell^{(e')}_u$ the value of $\ell_u$ at the beginning of the arrival of $e'$, we have $\ell^{(e')}_u + x_{e'} \leq 1$ due to \eqref{eq:ec:assume:effcovedge}. 
We thus derive
\[
    \Pr \Big[ \ell^{(e')}_u < Z_u \leq \ell^{(e')}_u + x_{e'} \Big]
    =
    \int_{\ell^{(e')}_u}^{\ell^{(e')}_u + x_{e'}} g(z) \, \dd z
    \geq
    g(0) \, x_{e'},
\]
where the inequality comes from the fact that $g$ is increasing, implying that
the probability that $u$ unconditionally marks $e' \in \deltaeff_e(u) \cap \deltaeff_e(v)$ is at least $p \, g(0) \, x_{e'}$.
We can therefore deduce that the probability that $u$ does not unconditionally mark any of $\deltaeff_e(u) \cap \deltaeff_e(v)$ is bounded from above by
\[
1 - \sum_{e' \in \deltaeff_e(u) \cap \deltaeff_e(v)} p \, g(0) \, x_{e'}
\leq \exp \left( -p \, g(0) \cdot \sum_{e' \in \deltaeff_e(u) \cap \deltaeff_e(v)} x_{e'} \right).
\]
Note also that any $e' \in \deltared_e(u) \cap \deltaeff_e(v)$ is marked by itself independently with probability $p \, g(0) \, x_{e'}$ by definition, implying that the probability that any of $\deltared_e(u) \cap \deltaeff_e(v)$ is not marked is bounded from above by
\[
	\prod_{e' \in \deltared_e(u) \cap \deltaeff_e(v)} (1 - p \, g(0) \, x_{e'}) \leq \exp \left( -p \, g(0) \cdot \sum_{e' \in \deltared_e(u) \cap \deltaeff_e(v)} x_{e'} \right).
\]
These observations together imply 
\begin{align*}
    \Pr[\calB_u] 
    &
    \leq \exp \left( -p \, g(0) \cdot \sum_{e' \in \delta_e(u) \cap \deltaeff_e(v)} x_{e'} \right).
\end{align*}
Since $\{\calB_u\}_{u \in N_e(v)}$ are mutually independent, 
we obtain that
\[
    \Pr[\text{Event~\eqref{ev:ec:n2b}}]
    = \prod_{u \in N_e(v)} \Pr[\calB_u]
    \leq \prod_{u \in N_e(v)} \exp \left( -p \, g(0) \cdot \sum_{e' \in \delta_e(u) \cap \deltaeff_e(v)} x_{e'} \right)
    = \e^{- p \, g(0) \, \ell_v},
\]
where we use the fact that $\ell_v = \sum_{u \in N_e(v)} \sum_{e' \in \delta_e(u) \cap \deltaeff_e(v)} x_{e'}$.
From this bound, we can see
\begin{equation} \label{eq:ec:factor}
    \Pr[v \text{ marks } e] 
    \leq \Pr[\calE]
    \leq \left( \int_{\ell_v}^{\ell_v + x_e} g(z) \dd z \right) \cdot \left( p + (1-p) \cdot \e^{-p \, g(0) \, \ell_v} \right).
\end{equation}

The next two propositions will be used to derive the factor $(\frac{\ln c}{p} + \frac{\varepsilon}{2})$ from \eqref{eq:ec:factor}. Recall the definition of $g$:
\[
    g(z) := \frac{\ln c}{p} \cdot \frac{1}{p + (1-p) \cdot c^{-z}} 
    \;\;\text{for every $z \in [0, 1]$,}
\]
where $p \in (0, 1]$ and $c = \frac{\e^{p^2}-1}{p} + 1 > 1$.
\begin{proposition} \label{lem:ec:eps}
    For any constant $\varepsilon > 0$, there exists a constant $\tau > 0$ satisfying that, for any $\ell \in [0, 1]$ and $y \in (0, \tau]$ such that $\ell + y \leq 1$, 
    \[
    \frac{1}{y} \cdot \int_\ell^{\ell+y} g(z) \dd z \leq g(\ell) + \frac{\varepsilon}{2}.
    \]
\end{proposition}
\begin{proof}
    Since $g$ is increasing, we have $ \frac{1}{y} \cdot \int_\ell^{\ell+y} g(z) \dd z \leq g(\ell+y)$. Therefore, it suffices to find $\tau$ that satisfies
    $
    g(\ell+y) - g(\ell) \leq \frac{\varepsilon}{2}
    $
    for every $y \in (0, \tau]$.
    Observe that
    \begin{align*}
        g(\ell+y) - g(\ell)
        &
        = \frac{\ln c}{p} \cdot \frac{(1 - p) (c^{-\ell} - c^{-(\ell + y)})}{(p + (1-p) c^{-(\ell + y)})(p + (1-p) c^{-\ell})} 
        \\ &
        \leq \frac{\ln c}{p} \cdot \frac{1-p}{(p + \frac{1-p}{c})^2} \cdot (c^\tau - 1).
    \end{align*}
    Therefore, by choosing any $\tau > 0$ such that
    \[
    \tau \leq \log_c \left( 1 + \frac{p(p + \frac{1-p}{c})^2}{2(1-p) \ln c} \cdot \varepsilon \right),
    \]
    the proposition follows.
\end{proof}

\begin{proposition} \label{lem:ec:factor}
    For any $\ell \in [0, 1]$, we have
    \[
    g(\ell) \cdot (p + (1-p) \cdot \e^{-p \cdot g(0) \cdot \ell}) = \frac{\ln c}{p}.
    \]
\end{proposition}
\begin{proof}
    Note that $g(0) = (\ln c)/p$, and hence,
    \[
        p + (1-p) \e^{-p \cdot g(0) \cdot \ell} = p + (1-p) \cdot c^{-\ell}.
    \]
    Multiplying by $g(\ell)$ on both sides completes the proof.
\end{proof}

By choosing $\tau > 0$ due to Proposition~\ref{lem:ec:eps}, \eqref{eq:ec:factor} implies
\begin{align*}
    \Pr[v \text{ marks } e]
    &
    \leq x_e \cdot \left( \frac{1}{x_e} \cdot \int_{\ell_v}^{\ell_v + x_e} g(z) \dd z \right) \cdot \left( p + (1-p) \cdot \e^{-p \, g(0) \, \ell_v} \right) 
    \\ & 
    \leq x_e \cdot \left( g(\ell_v) + \frac{\varepsilon}{2} \right) \cdot \left( p + (1-p) \cdot \e^{-p \, g(0) \, \ell_v} \right) 
    \\ &
    \leq \left( \frac{\ln c}{p} + \frac{\varepsilon}{2} \cdot \left( p + (1-p) \cdot \e^{-p \, g(0) \, \ell_v} \right) \right) \, x_e,
    \\ &
    \leq \left( \frac{\ln c}{p} + \frac{\varepsilon}{2} \right) \, x_e,
\end{align*}
where the second inequality is due to Proposition~\ref{lem:ec:eps}, and the second-to-last is due to Proposition~\ref{lem:ec:factor}.
This completes the proof of Lemma~\ref{lem:ec:main}.
\end{proof}

We are now ready to prove Theorem~\ref{thm:intro:edgecvr}.
\begin{proof} [Proof of Theorem~\ref{thm:intro:edgecvr}]
As we chose $p := 0.46$, we have $(\ln c)/p = g(0) < 0.8992$.
Therefore, due to Lemma~\ref{lem:ec:main}, for $\varepsilon = 0.001$, we can choose $\tau > 0$ such that every edge $e \in E$ is marked by an endpoint that $e$ effectively covers with probability at most $0.9 \, x_e$.
Recall that the rounding scheme inserts an edge $e$ into $C$ when $e$ is marked.
If $e$ effectively covers both endpoints, $e$ can be marked by its endpoints (but not by itself), leading to
\[
    \Pr[e \in C] 
    \leq \sum_{v \in e} \Pr[v \text{ marks } e]
    \leq 1.8 \, x_e.
\]
On the other hand, if $e$ effectively covers exactly one endpoint $v \in e$, then $e$ can be marked by either $v$ or itself (but not by the other endpoint than $v$), yielding that
\begin{align*}
    \Pr[e \in C] 
    & 
    \leq \Pr[v \text{ marks } e] + \Pr[e \text{ is marked by itself}]
    \leq \left(0.9 + p \, g(0) \right) \, x_e 
    \leq 1.8 \, x_e.
\end{align*}
The scheme never inserts $e$ into $C$ if it redundantly covers both endpoints.
\end{proof}
\section{\texorpdfstring{$(1+1/\e)$}{1+1/e}-Competitive Rounding Scheme for Bipartite Graphs under One-Sided Vertex Arrivals} \label{sec:varr}
In this section, we present an improved rounding scheme for the setting where $G = (U \cup V, E)$ is a simple bipartite graph, and at each round, the edges incident with a vertex in $v \in V$ are simultaneously revealed.
In what follows, we call each vertex in $U$ \emph{offline} and each vertex in $V$ \emph{online}.

\paragraph{Scheme Description}
We again denote by $C$ the edge cover output by the scheme, initially $C := \emptyset$.
As an initialization step, for each offline vertex $u \in U$, the scheme samples $Z_u$ uniformly at random independently from $(0, 1]$.
Now, in each round, the edges $\delta(v)$ incident to an online vertex $v \in V$ are revealed, as well as their solution values $\{x_{e}\}_{e \in \delta(v)}$.
For each incident edge $e = \{u, v\} \in \delta(v)$, if the amount by which $u$ is fractionally covered exceeds $Z_u$ for the first time in this round, the scheme selects $e$ into $C$.
Moreover, if $e$ covers $u$ (partially) redundantly, the scheme independently selects $e$ into $C$ with probability equal to the redundant amount of coverage due to $e$, i.e., $x_e - \max\{1 - \ell_u, 0\}$ where $\ell_u$ denotes the amount $u$ is covered at the beginning of this round.
However, it is possible that none of the edges are selected in this enumeration of $\delta(v)$; in this case, the scheme proceeds to a fallback process that randomly inserts one incident edge $e \in N(v)$ with probability proportional to $x_e$ (i.e., $\frac{x_e}{\sum_{e' \in \delta(v)} x_{e'}}$).
A pseudocode of this scheme is presented in Algorithm~\ref{alg:bipart}.

\begin{algorithm}[t]
    \caption{$(1+1/\e)$-competitive rounding scheme for bipartite graphs under one-sided vertex arrivals} \label{alg:bipart}
    \KwIn{$G = (U \cup V, E)$ and $x \in [0,1]^E$}
    \KwOut{An edge cover $C \subseteq E$}
    $C \gets \emptyset$\;
    \For{each offline vertex $u \in U$}{
        Sample $Z_u \overset{\mathrm{iid}}{\sim} \Unif (0, 1]$\;
    }
    \For{each online vertex $v \in V$ in the arrival order}{
        \For{each incident edge $e = \{u, v\} \in \delta(v)$}{
            $\ell_u \gets \sum_{e' \in \delta_e(u)} x_{e'}$\;
            \uIf{$\ell_u < Z_u \leq \min\{\ell_u + x_e, 1\}$}{
                $C \gets C \cup \{e\}$ \;
            }
            \uIf{$\ell_u + x_e > 1$} {
                $C \gets C \cup \{e\}$ with probability $x_e - \max \{1-\ell_u, 0\} $\;
            }
        }
        \uIf{$\delta(v) \cap C = \emptyset$} {
            Sample $e \in \delta(v)$ with probability 
            $\frac{x_e}{\sum_{e' \in \delta(v)} x_{e'}}$ \;
            $C \gets C \cup \{e\}$\;
        }
    }
    \Return{$C$}\;
\end{algorithm}

\paragraph{Analysis}
It is easy to see that the scheme returns a feasible integral edge cover $C$:
For each offline vertex $u \in U$, there must exist an edge $e$ such that $\sum_{e' \in \delta_e(u)} x_{e'} < Z_v \leq \min\{\sum_{e' \in \delta_e(u)} x_{e'} + x_e, 1\}$ since $x$ is feasible and $Z_v \in (0, 1]$.
For each online vertex $v \in V$, if it is not covered by the initial enumeration of $\delta(v)$, the scheme then randomly inserts one incident edge in the fallback process.

Let us now prove the competitive ratio of this scheme.
Note that the next lemma immediately shows Theorem~\ref{thm:intro:varr}
\begin{lemma}
    For every $e = \{u, v\} \in E$, $\Pr[e \in C] \leq (1 + 1/\e) \cdot x_e$. 
\end{lemma}
\begin{proof}
    Consider the round where $v$ arrives along with $\{x_{e'}\}_{e' \in \delta(v)}$.
    For each $e' = \{u', v\} \in \delta(v)$, let us denote
    \begin{itemize}
        \item by $\calA_{e'}$ the event that $\ell_{u'} < Z_{u'} \leq \min\{\ell_{u'} + x_{e'}, 1\}$;
        \item by $\calB_{e'}$ the event of $\mathrm{Bernoulli}(x_{e'} - \max\{ 1 - \ell_{u'}, 0\})$ given $\ell_{u'} + x_{e'} > 1$; and
        \item by $\calC_{e'}$ the event that $e' \in \delta(v)$ is sampled in the fallback process. 
    \end{itemize}
    We remark that $\{\calA_{e'}, \calB_{e'}\}_{e' \in \delta(v)}$ are mutually independent.
    Observe that, for every $e' \in \delta(v)$, the probability that none of $\calA_{e'}$ and $\calB_{e'}$ happen is exactly represented by $\Pr[\neg \calA_{e'} \wedge \neg \calB_{e'}] = (1 - \alpha_{e'})(1 - \beta_{e'})$ for some $\alpha_{e'}, \beta_{e'} \in [0, 1]$ such that $\alpha_{e'} + \beta_{e'} = x_{e'}$.
    Precisely, due to independence, the values of $\alpha$ and $\beta$ are determined as follows:
    \begin{itemize}
        \item $\alpha_{e'} = x_{e'}$ and $\beta_{e'} = 0$ if $\ell_{u'} + x_{e'} \leq 1$;        
        \item $\alpha_{e'} = 1 - \ell_{u'}$ and $\beta_{e'} = x_{e'} - 1 + \ell_{u'}$ if $\ell_{u'} < 1 < \ell_{u'} + x_{e'}$; and 
        \item $\alpha_{e'} = 0$ and $\beta_{e'} = x_{e'}$ if $\ell_{u'} \geq 1$.
    \end{itemize}
    This implies that
    \begin{equation*}
        \Pr[\calA_e \vee \calB_e] = 1 - \Pr[\neg \calA_e \wedge \neg \calB_e] \leq x_e.
    \end{equation*}
    Moreover, the probability that none of $\delta(v)$ are inserted into $C$ from the initial enumeration of $\delta(v)$ is 
    \[
        \prod_{e' \in \delta(v)} (1 - \alpha_{e'})(1 - \beta_{e'})
        \leq \exp \Big( -\sum_{e' \in \delta(v)} (\alpha_{e'} + \beta_{e'}) \Big)
        = \exp \Big( - \sum_{e' \in \delta(v)} x_{e'} \Big)
        \leq \e^{-1},
    \]
    where the second inequality comes from the feasibility of $x$.
    Conditioned on this event, the scheme selects $e$ with probability (at most) $x_e$, yielding that
    \begin{equation*}
        \Pr[\calC_e] \leq \frac{x_e}{\e}.
    \end{equation*}
    Since the scheme inserts $e$ into $C$ if and only if at least one of $\calA_e$, $\calB_e$, and $\calC_e$ happens, we derive from the above inequalities
    \[
        \Pr[e \in C] = \Pr[\calA_e \vee \calB_e \vee \calC_e] \leq \Big(1 + \frac{1}{\e} \Big) \, x_e,
    \]
    completing the proof.
\end{proof}
\section{Lower Bounds}
In this section, we discuss lower bounds on competitiveness of online rounding schemes for edge cover.
In the general case where $G$ is non-bipartite, recall that the standard LP relaxation is well-known to have an integrality gap of at least $4/3$ due to a length-3 cycle with $x_e := 1/2$ for every edge $e$.
\begin{corollary} \label{cor:lb:general}
    The competitive ratio of any online rounding scheme must be at least $4/3$.
\end{corollary}

We now consider the case where $G$ is bipartite.
We remark that the standard LP relaxation now has an integrality gap of $1$.
This implies that if we are informed of all the entries of a feasible fractional solution at the very beginning, we can round the fractional solution into an integral edge cover with no loss in cost.
However, in the online arrival model, it is impossible to round the given fractional solution losslessly in cost.
\begin{theorem} \label{thm:lb:earr}
    The competitive ratio of any online rounding scheme for bipartite graphs must be at least $7/6$.
\end{theorem}
\begin{proof}
    Consider two unweighted cycles $E_1$ and $E_2$ of length 6 on vertices $\{ v_1, v_2, v_3, v_4, v_5, v_6 \}$, defined as follows:
    \begin{itemize}
        \item $E_1 := \big\{ \{v_1, v_2\}, \{v_1, v_6\}, \{v_4, v_3\}, \{v_4, v_5\} \big\} \cup \big\{ \{v_2, v_3\}, \{v_5, v_6\} \big\}$;
        \item $E_2 := \big\{ \{v_1, v_2\}, \{v_1, v_6\}, \{v_4, v_3\}, \{v_4, v_5\} \big\} \cup \big\{ \{v_2, v_5\}, \{v_3, v_6\} \big\}$.
    \end{itemize}
    Observe that the incidences of $v_1$ and $v_4$ are the same across $E_1$ and $E_2$, while the incidences within the other vertices are different.
    The input fractional solution $x$ is composed of $1/2$ on every entry, which is indeed feasible to the LP relaxation.
    We now select one cycle $E$ between $E_1$ and $E_2$ uniformly at random, and first reveal the edges incident with $v_1$ and $v_4$ (i.e., $E_1 \cap E_2$) in an arbitrary order, followed by the remaining edges $E \setminus (E_1 \cap E_2)$.
    
    At the moment when all the common edges $E_1 \cap E_2$ are revealed, any rounding scheme must select (at least) one incident edge for each of $v_1$ and $v_4$, since the two vertices are fully covered.
    Let us assume for now that the rounding scheme has inserted $\{v_1, v_2\}$ and $\{v_4, v_5\}$ into the solution.
    In this case, note that the rounding scheme must output an edge cover of cardinality 4 if $E = E_1$ because $v_3$ (and $v_6$) should be covered by $\{v_3, v_2\}$ (and $\{v_6, v_5\}$, respectively) if $\{v_3, v_4\}$ (and $\{v_6, v_4\}$) has not been inserted.
    The other cases can be shown through a symmetric argument.

    The above argument implies that, with probability $1/2$, any rounding scheme should output an edge cover of size 4, while the size of $x$ is 3.
    Therefore, the competitive ratio of any rounding scheme should be at least $\frac{1}{2} + \frac{1}{2} \cdot \frac{4}{3} = \frac{7}{6}$.
\end{proof}

\begin{figure}[t]
\centering
\def\tikzpicscale{0.7}
\begin{subfigure}{0.3\textwidth}
\centering
\begin{tikzpicture}[scale=\tikzpicscale, transform shape]
    \foreach \i in {1,2,3} {
        \node[vertex] (a\i) at ({90+120*(\i-1)}:2) {};
    }
    \draw[edge] (a1)--(a2)--(a3)--(a1);
\end{tikzpicture}
\caption{Corollary~\ref{cor:lb:general}}
\end{subfigure}
\hfill
\begin{subfigure}{0.3\textwidth}
\centering
\begin{tikzpicture}[scale=\tikzpicscale, transform shape]
    \foreach \i in {1,...,6} {
        \node[vertex] (v\i) at ({90-60*(\i-1)}:2) {$v_\i$};
    }
    \draw[edge] (v6)--(v1)--(v2);
    \draw[edge] (v3)--(v4)--(v5);

    \draw[edge, draw=red, densely dotted] (v2)--(v3);
    \draw[edge, draw=red, densely dotted] (v5)--(v6);

    \draw[edge, draw=blue, dashed] (v2)--(v5);
    \draw[edge, draw=blue, dashed] (v3)--(v6);
\end{tikzpicture}
\caption{Theorem~\ref{thm:lb:earr}}
\end{subfigure}
\hfill
\begin{subfigure}{0.3\textwidth}
\centering
\begin{tikzpicture}[scale=\tikzpicscale, transform shape]
    \foreach \i in {1,2,3,4} {
        \node[vertex] (v\i) at ({90-45*(2*\i-2)}:2) {$v_\i$};
    }
    \foreach \i in {1,2,3,4} {
        \node[offvertex] (u\i) at ({90-45*(2*\i-1)}:2) {$u_\i$};
    }

    \draw[edge] (u4)--(v1)--(u1);
    \draw[edge] (u2)--(v3)--(u3);

    \draw[edge, draw=red, densely dotted] (u1) to[bend left=30] (v2);
    \draw[edge, draw=red, densely dotted] (v2)--(u2);
    \draw[edge, draw=red, densely dotted] (u3)--(v4);
    \draw[edge, draw=red, densely dotted] (v4) to[bend left=30] (u4);

    \draw[edge, draw=blue, dashed] (u1) to[bend right=30] (v2);
    \draw[edge, draw=blue, dashed] (v2)--(u3);
    \draw[edge, draw=blue, dashed] (u2)--(v4);
    \draw[edge, draw=blue, dashed] (v4) to[bend right=30] (u4);
\end{tikzpicture}
\caption{Theorem~\ref{thm:lb:varr}}
\end{subfigure}
\caption{Lower bounds on the competitiveness of online rounding schemes for edge cover. For Theorems~\ref{thm:lb:earr} and~\ref{thm:lb:varr}, $E_1$ is shown as red dotted edges, while $E_2$ is shown as blue dashed edges. All edges have fractional solution value $1/2$.}
\end{figure}

Adapting the above lower bound for bipartite graphs, we can also obtain a lower bound for the setting considered in Section~\ref{sec:varr}.
\begin{theorem} \label{thm:lb:varr}
    The competitive ratio of any online rounding scheme for bipartite graphs under one-sided vertex arrivals must be at least $9/8$.
\end{theorem}
\begin{proof}
    We now consider two unweighted cycles $E_1$ and $E_2$ of length 8 on vertices $U := \{u_1, u_2, u_3, u_4\}$ and $V := \{v_1, v_2, v_3, v_4\}$, defined as follows:
    \begin{itemize}
        \item $E_1 := \big\{ \{u_1, v_1\}, \{u_4, v_1\}, \{u_2, v_3\}, \{u_3, v_3\} \big\} \cup \big\{ \{u_1, v_2\}, \{u_2, v_2\}, \{u_3, v_4\}, \{u_4, v_4\} \big\}$;
        \item $E_2 := \big\{ \{u_1, v_1\}, \{u_4, v_1\}, \{u_2, v_3\}, \{u_3, v_3\} \big\} \cup \big\{ \{u_1, v_2\}, \{u_3, v_2\}, \{u_2, v_4\}, \{u_4, v_4\} \big\}$.
    \end{itemize}
    The input fractional solution $x$ is again comprised of $1/2$ on every entry.

    We select one cycle $E$ between $E_1$ and $E_2$ uniformly at random, and reveal the edges incident with $v_1$, $v_3$, $v_2$, and $v_4$ in this order.
    At the moment when $v_1$ and $v_3$ have arrived, any rounding scheme must have selected (at least) one incident edge for each of $v_1$ and $v_3$.
    Let us say that $\{u_1, v_1\}$ and $\{u_2, v_3\}$ are chosen in the solution.
    We can then see that, if $E = E_1$, the rounding scheme must output an edge cover of size 5 to cover the remaining vertices $u_3$, $u_4$, $v_2$, and $v_4$.
    A symmetric argument would again prove the other cases.

    We thus conclude that, with probability $1/2$, any rounding scheme should output an edge cover of size at least 5 while the size of $x$ is 4, implying that the competitive ratio must be at least $\frac{1}{2} + \frac{1}{2} \cdot \frac{5}{4} = \frac{9}{8}$.
\end{proof}

\begin{remark}
    If we select a cycle adaptively to the intermediate selection of a rounding scheme, we can instead derive $4/3$ and $5/4$, respectively.
    Therefore, the competitive ratios of \emph{deterministic} online rounding schemes must be at least $4/3$ and $5/4$ under the settings of Theorems \ref{thm:lb:earr} and \ref{thm:lb:varr}, respectively.
\end{remark}

\section*{Acknowledgment}
The authors thank Marek Adamczyk for valuable discussions on existing results for contention resolution problems. The authors also thank the anonymous reviewers for their insightful feedback, in particular for suggesting that a folklore offline rounding scheme can be adapted to an online rounding scheme for general set cover under subset arrivals, and for identifying a gap in the analysis in an earlier version of this paper.
This work was supported by NCN grant number 2020/39/B/ST6/01641.

\bibliographystyle{plain}
\bibliography{ref}

\end{document}